\documentclass{llncs}
\usepackage{amssymb,amsmath}
\usepackage{graphicx,psfrag}
\usepackage{color}

\usepackage{ulem}

\begin{document}

\newcommand{\red}[1]{\textcolor{red}{#1}} \newcommand{\?}{{\bf ???}}

\definecolor{axcolor}{rgb}{.3,0,.3}
\newcommand{\Triv}{\mathit{Triv}}
\newcommand{\timed}{\mathsf{time}}
\newcommand{\sqspace}{\mathsf{space}^2}
\newcommand{\sqlength}{\mathsf{length}^2}
\newcommand{\sqspeed}{\mathsf{speed}^2}
\newcommand{\sqvelocity}{\mathsf{velocity}^2}
\newcommand{\ls}{\mathsf{c}}
\newcommand{\vo}{\bar o}
\newcommand{\vr}{\bar r}
\newcommand{\vx}{\bar x}
\newcommand{\vy}{\bar y}
\newcommand{\vz}{\bar z}
\newcommand{\vv}{\bar v}
\newcommand{\vu}{\bar u}
\newcommand{\de}{:=}
\newcommand{\defiff}{\ \stackrel{\;def}{\Longleftrightarrow}\ }
\newcommand{\IOb}{\ensuremath{\mathsf{IOb}}} 
\newcommand{\IB}{\ensuremath{\mathsf{IB}}} 
\newcommand{\Ob}{\ensuremath{\mathsf{Ob}}} 
\newcommand{\B}{\ensuremath{\mathit{B}}} 
\newcommand{\Ph}{\ensuremath{\mathsf{Ph}}} 
\newcommand{\Q}{\ensuremath{\mathit{Q}}} 
\newcommand{\W}{\ensuremath{\mathsf{W}}} 
\newcommand{\ev}{\ensuremath{\mathsf{ev}}} 
\newcommand{\dom}{\ensuremath{\mathit{Dom}\,}} 
\newcommand{\ran}{\ensuremath{\mathit{Ran}\,}} 
\newcommand{\ax}[1]{\textcolor{axcolor}{\ensuremath{\mathsf{#1}}}} 
\newcommand{\wl}{\ensuremath{\mathsf{wl}}}
\newcommand{\lc}{\ensuremath{\mathsf{lc}}}
\newcommand{\w}{\ensuremath{\mathsf{w}}}

\titlerunning{Hypercomputation in SR} \title{Existence of Faster Than
  Light Signals Implies Hypercomputation Already in Special
  Relativity\thanks{This research is supported by the Hungarian
    Scientific Research Fund for basic research grants No.~T81188 and
    No.~PD84093.}}  \author{P\'eter N\'emeti \and Gergely Sz\'ekely}
\institute{ Alfr\'ed R\'enyi Institute of Mathematics\\ Mailing
  address: POB 127, H-1364 Budapest, Hungary
  \\ \email{nemeti.peter@renyi.mta.hu, szekely.gergely@renyi.mta.hu}}
\maketitle

\begin{abstract} 
Within an axiomatic framework, we investigate the possibility of
hypercomputation in special relativity via faster than light
signals. We formally show that hypercomputation is theoretically
possible in special relativity if and only if there are faster than
light signals.
\smallskip

{\bf Keywords:} relativistic computation, special relativity, faster than light signals
\end{abstract}

\section{Introduction}
The theory of relativistic hypercomputation (i.e., the investigation
of relativity theory based physical computational scenarios which are
able to solve non-Turing-computable problems) has an extensive
literature and it is investigated by several researchers in the past
decades, see, e.g., \cite{ANN}, \cite{nemeti-dgy},
\cite{Earman-Norton}, \cite{Etesi-Nemeti}, \cite{Hogarth92},
\cite{Manchak}. For an overview of different approaches to
hypercomputation, see, e.g., \cite{Stannett}.

It is well-known that hypercomputation is not possible in special
relativity in the usual sense (i.e., the sense of Malament--Hogarth
spacetimes), see, e.g., \cite{Hogarth92}. In this paper, we show that
it is possible to perform relativistic hypercomputation via
  ordinary computers (Turing machines) in special relativity if there
are faster than light (FTL) signals, e.g., particles. We will also
show that there have to be FTL signals if relativistic
hypercomputation is possible in special relativity (via Turing
machines), see Thm.\ref{thm-hc}.

It is interesting in and of itself to investigate the (logical)
consequences of the assumption that FTL objects exist, independently
of the question whether they really exist or not in our actual
physical universe. Logic based axiomatic investigations typically aim
for describing all the theoretically possible universes and not just
our actual one. Moreover, so far we have not excluded the
  possibility of the existence of FTL entities in our actual universe;
  and from time to time there appear theories and experimental results
  suggesting the existence of FTL objects. Recently, the OPERA
  experiment, see \cite{opera}, raised the interest in the possibility
  of FTL particles.

Contrary to the common belief, the existence of FTL particles does
not lead to a logical contradiction within special relativity. For a
formal axiomatic proof of this fact, see \cite{FTLconsSR}. However, it
is interesting to note that, in contrast with this result, the
impossibility of the existence of FTL inertial {\it observers} follows from
special relativity, see, e.g., \cite{Synthese}.

The investigation of FTL motion in relativity theory goes back (at
least) to Tolman, see, e.g., \cite[p.54-55]{tolman}.  Since then a
great many works dealing with FTL motion have appeared in the
literature, see, e.g., \cite{matolcsi-ftl}, \cite{mittelstaedt-ftl},
\cite{recami-ftl}, \cite{recami-ftl2}, \cite{RFG-ftl},
\cite{selleri-ftl}, \cite{weinstein-ftl} to mention only a few.

\section{Hypercomputation in SR}
\label{sec-hc}

It is well-known that we can send information back to the past if
there are FTL particles, see, e.g., \cite{FTLconsSR},
\cite[p.54-55]{tolman}.  It is natural to try using this possibility
to design computers with greater computational power.  We will show
that uniformly accelerated relativistic computers can compute beyond
the Church--Turing barrier via using FTL signals. In this section, we
show this fact informally.  In Sect.\ref{sec-HC}, we reconstruct our
informal ideas of this section within an axiomatic theory of special
relativity extended with accelerated observers.

Our first observation is that if we can send out an FTL signal with a
certain speed, we also have to be able to send out arbitrarily fast
signals, by the principle of relativity.  Prop.\ref{prop-nolim} is a
formal statement of this observation.  To informally justify this
statement, let us assume that we can send out an FTL signal by a
certain experiment, say with speed $1.01\mathrm{c}$. According to
special relativity, for any FTL speed, say $10^{10}\mathrm{c}$, there
is a inertial reference frame (moving relative to our frame) according
to which our signal moves with this speed. By the principle of
relativity, inertial frames are experimentally indistinguishable, see
\cite[\S 5, pp.149-159]{friedman}, \cite{galileo},
\cite[pp.176-178]{taylor-wheeler}. So the experiment which is
configured in our reference frame as our original experiment is seen
by this moving inertial frame as yielding an FTL signal moving with speed
$10^{10}\mathrm{c}$ in our frame.  Therefore, in our (or any other
inertial) reference frame, it is possible to send out an FTL signal
with any speed.

Let us see the construction of our special relativistic
hypercomputer. Let the computer be accelerated uniformly with respect
to an inertial observer, see Fig.\ref{fig-hc}.\footnote{In relativity
  theory, uniform acceleration means motion along a hyperbola
  (according to inertial observers), see, e.g., \cite[\S 3.8,
    pp.37-38]{dinverno}, \cite[\S 6]{MTW}, \cite[\S 12.4,
    pp.267-272]{Rindler}.}  There is an event $O$ with the following
property: any event $E$ on the worldline of our uniformly accelerated
computer is simultaneous with $O$, according to the inertial observer
comoving with the computer at $E$, see, e.g.,
\cite[Fig.6.4, p.173]{MTW}, \cite[Fig.5.13, p.152]{petkov}.

\begin{figure}
\begin{center}
\psfrag{S}[br][br]{\textcolor{magenta}{FTL signal}} 
\psfrag{k}[tl][tl]{\textcolor{blue}{Computer}}
\psfrag{m}[tl][tl]{\textcolor{green}{Programmer}}
\psfrag{E}[tr][tr]{$E$} 
\psfrag{O}[tr][tr]{$O$}
\psfrag{M}[tr][tr]{$M$}
\psfrag{Cm}[tl][tl]{Comoving observer at 2}
\psfrag{Sim}[tl][tl]{\parbox{120pt}{Simultaneity of the comoving observer at 2}}
\psfrag{Sim1}[tl][tl]{\parbox{120pt}{Simultaneity of the comoving observer at 1}}
\psfrag{0}[tl][tl]{0} 
\psfrag{1}[tl][tl]{1}  
\psfrag{2}[tl][tl]{2} 
\psfrag{3}[tl][tl]{3}
\includegraphics[keepaspectratio,width=0.7\textwidth]{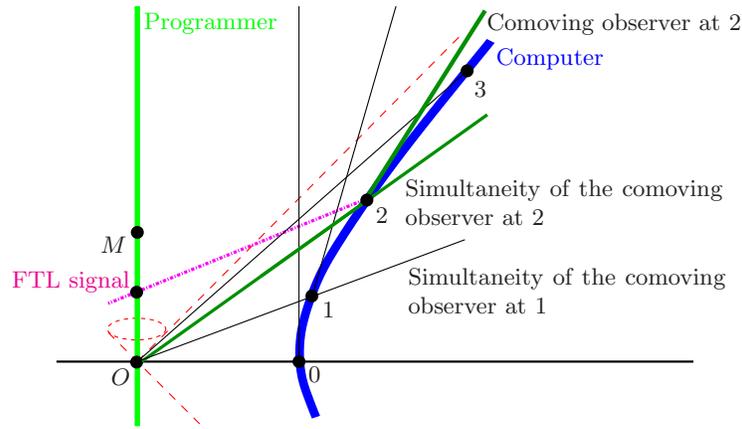}
\caption{Illustration of hypercomputation via FTL particles}
\label{fig-hc}
\end{center}
\end{figure}

Now let us show that this configuration can be used to decide
non-Turing-computable questions if there are FTL signals.  Let us set
the computer to work on some recursively enumerable but non
Turing-computable problem, say the decision problem for the
consistency of ZF set theory; the computer enumerates one by one all
the consequences of ZF. Let us fix an event $M$ on the worldline of
the programmer which is later than $O$ according to him. Now, if the
computer finds a contradiction, let it send out a fast enough signal
which reaches the programmer before event $M$. Such signal exists
since, by our first observation, the computer can send out a signal
which is arbitrarily fast with respect to his coordinate system (i.e.,
any half line in the ``upper'' half space determined by the comoving
observer's simultaneity can be the worldline of the signal).
Therefore, if the programmer receives a signal between events $O$ and
$M$, he knows that ZF is inconsistent; and if there is no signal
between $M$ and $O$, he knows that the computer has not found any
contradiction, so after event $M$ the programmer can conclude that there
is no contradiction in ZF set theory.  The same way, by this thought
experiment using FTL signals, we can decide (experimentally) any
recursively enumerable set of numbers.

If there are no FTL signals, then the whole computation has to happen
in the causal past of the event when the programmer learns the result
of computation. However, in special relativity, the computer remaining
within the causal past of any event has only finite time to compute by
the twin paradox theorem.  That is why hypercomputation is not
possible in special relativity without FTL signals. This argument is
also the basis of proving that Minkowski spacetime is not a
Malament--Hogarth spacetime.

\section{The Language of Our Axiom Systems}
\label{lang-s}
To formalize the result of Sect.\ref{sec-hc}, we need an axiomatic
theory of special relativity extended with accelerated observers. To
introduce any axiomatic theory, first we have to fix the set of basic
symbols of the theory, i.e., what objects and relations between them
we will use as basic concepts.

Here we will use the following two-sorted\footnote{That our theory is
  two-sorted means only that there are two types of basic objects
  (bodies and quantities) as opposed to, e.g., Zermelo--Fraenkel set
  theory where there is only one type of basic objects (sets).}
language of first-order logic pa\-ram\-e\-trized by a natural number $d\ge
2$ representing the dimension of spacetime:
\begin{equation*}
\{\, \B,\Q\,; \Ob, \IOb, \Ph,+,\cdot,\le,\W\,\},
\end{equation*}
where $\B$ (bodies) and $\Q$ (quantities) are the two sorts,
$\Ob$ (observers), $\IOb$ (inertial observers) and $\Ph$ (light
signals) are one-place relation symbols of
sort $\B$, $+$ and $\cdot$ are two-place function symbols of sort
$\Q$, $\le$ is a two-place relation symbol of sort $\Q$, and $\W$ (the
worldview relation) is a $d+2$-place relation symbol the first two
arguments of which are of sort $\B$ and the rest are of sort $\Q$.

Relations $\Ob(o)$, $\IOb(m)$ and $\Ph(p)$ are translated as
``\textit{$o$ is an observer},'' ``\textit{$m$ is an inertial
  observer},'' and ``\textit{$p$ is a light signal},''
respectively. To speak about coordinatization, we translate
$\W(k,b,x_1,x_2,\ldots,x_d)$ as ``\textit{body $k$ coordinatizes body
  $b$ at space-time location $\langle x_1, x_2,\ldots,x_d\rangle$},''
(i.e., at space location $\langle x_2,\ldots,x_d\rangle$ and instant
$x_1$).

{\bf Quantity terms} are the variables of sort $\Q$ and what can be
built from them by using operations $+$ and $\cdot$, {\bf body terms} are
only the variables of sort $\B$.  Relations $\Ob(o)$, $\IOb(m)$,
$\Ph(p)$, $\W(m,b,x_1,\ldots,x_d)$, $x=y$, and $x\le y$ where $o$,
$m$, $p$, $b$, $x$, $y$, $x_1$, \ldots, $x_d$ are arbitrary terms of
the respective sorts are so-called {\bf atomic formulas} of our
first-order logic language.  {\bf Formulas} are built up from these
atomic formulas by using the logical connectives \textit{not}
($\lnot$), \textit{and} ($\land$), \textit{or} ($\lor$),
\textit{implies} ($\rightarrow$), \textit{if-and-only-if}
($\leftrightarrow$) and the quantifiers \textit{exists} ($\exists$)
and \textit{for all} ($\forall$).

To make them easier to read, we omit the outermost universal
quantifiers from the formalizations of our axioms, i.e., all the free
variables are universally quantified.

We use the notation $\Q^d$ for the set of all $d$-tuples of elements
of $\Q$. If $\vx\in \Q^d$, we assume that $\vx=\langle
x_1,\ldots,x_d\rangle$, i.e., $x_i$ denotes the
$i$-th component of the $d$-tuple $\vx$. Specially, we write $\W(m,b,\vx)$ in
place of $\W(m,b,x_1,\dots,x_d)$, and we write $\forall \vx$ in place
of $\forall x_1\dots\forall x_d$, etc.

We use first-order set theory as a meta theory to speak about model
theoretical terms, such as models.  The {\bf models} of this language
are of the form
\begin{equation*}
{\mathfrak{M}} = \langle \B, \Q;\Ob_\mathfrak{M},
\IOb_\mathfrak{M},\Ph_\mathfrak{M},+_\mathfrak{M},\cdot_\mathfrak{M},\le_\mathfrak{M},\W_\mathfrak{M}\rangle,
\end{equation*}
where $\B$ and $\Q$ are nonempty sets, $\Ob_\mathfrak{M}$,
$\IOb_\mathfrak{M}$ and $\Ph_\mathfrak{M}$ are unary relations on
$\B$, $+_\mathfrak{M}$ and $\cdot_\mathfrak{M}$ are binary functions
and $\le_\mathfrak{M}$ is a binary relation on $\Q$, and
$\W_\mathfrak{M}$ is a relation on $\B\times \B\times
\Q^d$.  Formulas are interpreted in $\mathfrak{M}$
in the usual way. For the precise definition of the syntax and
semantics of first-order logic, see, e.g., \cite[\S 1.3]{CK}, \cite[\S 2.1, \S
  2.2]{End}. 

\section{Axioms of Special Relativity}
\label{ax-s}

Let us recall some of our axioms for special relativity.
Our first axiom states some basic properties of addition,
multiplication and ordering true for real numbers.
\begin{description}
\item[\underline{\ax{AxOField}:}]
 The quantity part $\langle \Q,+,\cdot,\le \rangle$ is an ordered field, i.e.,
\begin{itemize}
\item  $\langle\Q,+,\cdot\rangle$ is a field in the sense of abstract
algebra; and
\item 
the relation $\le$ is a linear ordering on $\Q$ such that  
\begin{itemize}
\item[i)] $x \le y\rightarrow x + z \le y + z$ and 
\item[ii)] $0 \le x \land 0 \le y\rightarrow 0 \le xy$
holds.
\end{itemize}
\end{itemize}
\end{description}

In the next axiom, we will use the concepts of time difference and
spatial distance.  The {\bf time difference} of coordinate points
$\vx,\vy\in\Q^d$ is defined as:
\begin{equation*}
\timed(\vx,\vy)\de x_1-y_1. 
\end{equation*}
To speak about the spatial distance of any two coordinate points, we
have to use squared distance since it is possible that the distance of
two points is not amongst the quantities, e.g., the distance of points
$\langle 0,0\rangle$ and $\langle 1,1\rangle$ is $\sqrt{2}$. So in the
field of rational numbers, $\langle 0,0\rangle$ and $\langle
1,1\rangle$ do not have distance just squared distance.  Therefore, we
define the {\bf squared spatial distance} of $\vx,\vy\in\Q^d$ as:
\begin{equation*}
  \sqspace(\vx,\vy)\de (x_2-y_2)^2+\ldots+(x_d-y_d)^2.
\end{equation*}

Our next axiom is the key axiom of our axiom system of special
relativity.  This axiom is the
outcome of the Michelson-Morley experiment, and it has been
continuously tested ever since then.  Nowadays it is tested by GPS
technology.
\begin{description}
\item[\underline{\ax{AxPh:}}] For any inertial observer, the speed of
  light is the same everywhere and in every direction (and it is
  finite). Furthermore, it is possible to send out a light signal in
  any direction everywhere:
\begin{multline*}
 \IOb(m)\rightarrow \exists c_m \Big( c_m>0\land \forall \vx\vy
    \Big[\sqspace(\vx,\vy)=
   c_m^2\cdot\timed(\vx,\vy)^2\\\leftrightarrow \exists p \big[ \Ph(p)\land \W(m,p,\vx)\land
   \W(m,p,\vy)\big] \Big]\Big) .
\end{multline*}
\end{description}

Let us note here that \ax{AxPh} does not require (by itself) that the
speed of light is the same for every inertial observer.  It requires
only that the speed of light according to a fixed inertial observer is
a positive quantity which does not depend on the direction or the
location. However, by \ax{AxPh}, we can define the {\bf speed of
  light} according to inertial observer $m$ as the following binary
relation:
\begin{multline*}
\ls(m,v)\defiff v>0  \land \forall \vx\vy\Big[
\exists p  \big[\Ph(p)\land \W(m,p,\vx)\land \W(m,p,\vy)\big]\\
\rightarrow \sqspace(\vx,\vy)= v^2\cdot\timed(\vx,\vy)^2\Big]. 
\end{multline*}
By \ax{AxPh}, there is one and only one speed $v$ for every inertial
observer $m$ such that $\ls(m,v)$ holds. From now on, we will denote
this unique speed by $\ls_m$.

Our next axiom connects the worldviews of different inertial observers
by saying that they coordinatize the same ``external" reality
(the same set of events).  By the {\bf event} occurring for observer
$m$ at coordinate point $\vx$, we mean the set of bodies $m$
coordinatizes at $\vx$:
\begin{equation*}
\ev_m(\vx)\de\{ b : \W(m,b,\vx)\}.
\end{equation*}
\begin{description}
\item[\underline{\ax{AxEv}:}]
All inertial observers coordinatize the same set of events:
\begin{equation*}
 \IOb(m)\land\IOb(k)\rightarrow \exists \vy\, \forall
 b\big[\W(m,b,\vx)\leftrightarrow\W(k,b,\vy)\big].
\end{equation*}
\end{description}
From now on, we will abbreviate the
subformula $\forall b\big[ \W(m,b,\vx)\leftrightarrow\W(k,b,\vy)\big]$
of \ax{AxEv} to $\ev_m(\vx)=\ev_k(\vy)$. The next two axioms are only simplifying ones.
\begin{description}
\item[\underline{\ax{AxSelf}:}]
Any inertial observer is stationary relative to himself:
\begin{equation*}
\IOb(m)\rightarrow \forall \vx\big[\W(m,m,\vx) \leftrightarrow x_2=\ldots=x_d=0\big].
\end{equation*}
\end{description}
Our last axiom on inertial observers is a symmetry axiom saying that
they use the same units of measurement.
\begin{description}
\item[\underline{\ax{AxSymD}:}]
Any two inertial observers agree as to the spatial distance between
two events if these two events are simultaneous for both of them.
Furthermore, the speed of light is 1 for all observers:
\begin{multline*}
\IOb(m)\land\IOb(k) \land x_1=y_1\land
x'_1=y'_1\land \ev_m(\vx)=\ev_k(\vx')\\ \land
\ev_m(\vy)=\ev_k(\vy')\rightarrow \sqspace(\vx,\vy)=\sqspace(\vx',\vy'),
\mbox{ and }\\
\IOb(m)\rightarrow\exists
p\big[\Ph(p)\land\W(m,p,0,\ldots,0)\land\W(m,p,1,1,0,\ldots,0)\big].
\end{multline*}
\end{description}
Our axiom system \ax{SpecRel} is the collection of the five simple
axioms above:
\begin{equation*}
\ax{SpecRel} \de \{\ax{AxOField}, \ax{AxPh},\ax{AxEv},\ax{AxSelf},
\ax{AxSymD}\}.
\end{equation*}
 
To show that \ax{SpecRel} captures the kinematics of special
relativity, let us introduce the {\bf worldview transformation} between
observers $m$ and $k$ (in symbols, $\w_{mk}$) as the binary relation
on $\Q^d$ connecting the coordinate points where $m$ and $k$
coordinatize the same (nonempty) events:
\begin{equation*}
\w_{mk}(\vx,\vy)\defiff \ev_m(\bar
x)=\ev_k(\vy)\neq\emptyset.
\end{equation*}

Map $P:\Q^d\rightarrow\Q^d$ is called a Poincar\'e transformation iff
it is an affine bijection such that, for all $\vx,\vy,\vx',\vy'\in\Q^d$
for which $P(\vx)=\vx'$ and $P(\vy)=\vy'$,
\begin{equation*}
\timed(\vx,\vy)^2-\sqspace(\vx,\vy)=\timed(\vx',\vy')^2-\sqspace(\vx',\vy').
\end{equation*}

Thm.\ref{thm-poi} shows that our streamlined axiom system
\ax{SpecRel} perfectly captures the kinematics of special relativity
since it implies that the worldview transformations between inertial
observers are the same as in the standard non-axiomatic
approaches. For the proof of Thm.\ref{thm-poi}, see \cite{wnst}.
\begin{theorem}\label{thm-poi}
Let $d\ge3$. Assume \ax{SpecRel}. Then $\w_{mk}$ is a Poincar\'e
transformation if $m$ and $k$ are inertial observers.
\end{theorem}
The so-called {\bf worldline} of body $b$ according to observer $m$ is
defined as:
\begin{equation*}
\wl_m(b)\de\{ \vx: \W(m,b,\vx)\}.
\end{equation*}
\begin{corollary}\label{cor-line}
Let $d\ge3$. Assume \ax{SpecRel}. The $\wl_m(k)$ is a straight line if
$m$ and $k$ are inertial observers.
\end{corollary}

To extend \ax{SpecRel} to accelerated observers, we need further
axioms. We connect the worldviews of accelerated and inertial
observers by the next axiom.
\begin{description}
\item[\underline{\ax{AxCmv}:}] At each moment of its world-line, each
  observer coordinatizes the nearby world for a short while as an
  inertial observer does.
\end{description}
Axiom \ax{AxCmv} is captured by formalizing the following statement:
at each point of the worldline of an observer there is an inertial
comoving observer such that the derivative of the worldview
transformation between them is the identity map, see, e.g.,
\cite{Synthese} \cite[\S 6]{SzPhd} for details.  We will also use the
generalized (localized) versions of axioms \ax{AxEv} and \ax{AxSelf}
of \ax{SpecRel} assumed for every observer.
\begin{description}\label{axev-}
\item[\underline{\ax{AxEv^-}:}] Observers coordinatize all the events
  in which they participate:
\begin{equation*}
\Ob(k)\land \W(m,k,\vx)\rightarrow\exists \vy\enskip
\ev_m(\vx)=\ev_k(\vy).
\end{equation*}
\end{description}

\begin{description}
\item[\underline{\ax{AxSelf^-}:}] In his own worldview, the worldline
  of any observer is an interval of the time axis containing all the
  coordinate points of the time axis where the observer coordinatizes
  something:
\begin{multline*}
\big[\W(m,m,\vx)\rightarrow x_2=\ldots= x_d =0\big]
\land\\ \big[\W(m,m,\vy)\land\W(m,m,\vz)\land
y_1<t<z_1\rightarrow \W(m,m,t,0,\ldots,0)\big]
\land\\ \exists
b \,\big[\W(m,b,t,0,\ldots,0) \rightarrow \W(m,m,t,0,\ldots,0)\big].
\end{multline*}
\end{description}
Let us add these three axioms to \ax{SpecRel} to get a theory of
accelerated observers:
\begin{equation*}
\ax{AccRel_0} \de \ax{SpecRel}\cup\{\ax{AxCmv},\ax{AxEv^-},\ax{AxSelf^-}\}.
\end{equation*}
Since \ax{AxCmv} ties the behavior of accelerated observers to the
inertial ones and \ax{SpecRel} captures the kinematics of special
relativity perfectly by Thm.\ref{thm-poi}, it is quite natural to
think that \ax{AccRel_0} is a theory strong enough to prove the most
fundamental theorems about accelerated observers.  However,
\ax{AccRel_0} does not even imply the most basic predictions of
relativity theory about accelerated observers, such as the twin
paradox. Moreover, it can be proved that even if we add the whole
first-order logic theory of real numbers to \ax{AccRel_0} is not enough
to get a theory that implies (predicts) the twin paradox, see, e.g.,
\cite{twp}, \cite[\S 7]{SzPhd}.

In the models of \ax{AccRel_0} in which the twin paradox is not true,
there are some definable gaps in $\W$. Our next assumption excludes
these gaps.
\begin{description}
\label{p-cont}
\item[\underline{\ax{CONT}:}] Every parametrically definable, bounded
  and nonempty subset of $\Q$ has a supremum (i.e., least upper bound)
  with respect to $\le$.
\end{description}
In \ax{CONT}, ``definable'' means ``definable in the language of
\ax{AccRel}, parametrically.''  \ax{CONT} is Tarski's first-order
logic version of Hilbert's continuity axiom in his axiomatization of
geometry fitted to the language of
\ax{AccRel}. For a precise formulation of \ax{CONT}, see
\cite[p.692]{twp} or \cite[\S 10.1]{SzPhd}.  When $\Q$ is the ordered
field of real numbers, \ax{CONT} is automatically true.

Let us extend \ax{AccRel_0} with axiom schema \ax{CONT}:
\begin{equation*}
\ax{AccRel}:=\ax{AccRel_0}\cup\ax{CONT}.
\end{equation*}
It can be proved that \ax{AccRel} implies the twin paradox, see
\cite{twp}, \cite[\S 7.2]{SzPhd}.

That \ax{CONT} requires the existence of supremum only for sets
definable in \ax{AccRel} instead of every set is important because it
makes our postulate closer to the physical/empirical level. This is
true because \ax{CONT} does not speak about ``any fancy subset'' of
the quantities, but just about those ``physically meaningful'' sets
which can be defined in the language of our (physical) theory.

Let us now introduce some auxiliary axioms we will use here but not
listed so far.  To do so, let us call a linear bijection of $\Q^d$
{\bf trivial transformation} if leaves the time components (i.e.,
first coordinates) of coordinate points unchanged and it fixes the
points of the time axis, i.e., the set of trivial transformation is:
\begin{multline*}
\Triv:=\{\,T: T \mbox{ is a linear bijection of } \Q^d,\\ T(\vy)_1=\vy_1
\mbox{ and } T(\vx)=\vx \mbox{ if } \vx_s=\vo\,\},
\end{multline*} 
where $\vo$ denotes the {\bf origin}, i.e., coordinate point $\langle
0,\ldots,0\rangle$.

\begin{description}
\item[\underline{\ax{AxThExp^\#}:}] Inertial observers can move with
  any speed less than the speed of light and new inertial reference
  frames can be constructed from other inertial reference frames by
  transforming them by trivial transformations and translations along
  the time axis:\footnote{Since linear bijections of $\Q^d$ can be
    represented by a matrix of $d\times d$ quantities, the
    quantification $\forall T\enskip T\in\Triv$ in \ax{AxTheExp^\#}
    can easily turned into a quantification over quantities.}

\begin{multline*}
\exists h \, \IOb(h)\land
\Big[\IOb(m)\land \sqspace(\vx,\vy)<\ls_m^2\cdot\timed(\vx,\vy)^2
  \\\land T\in\Triv \rightarrow \exists km'
  \big[\IOb(k)\land\IOb(m')\land \W(m,k,\vx)\land\W(m,k,\vy)\\\land
    \ev_m(\vx)=\ev_k(\vo)\land \w_{mm'}=T\big]\Big].
\end{multline*}
\end{description}
The following axiom is a consequence of the principle of
relativity. See \cite{MPhd}, \cite{FTLconsSR} for a formalization of
the principle of relativity in our first-order logic language.

\begin{description}
\item[\underline{\ax{AxVel}:}] If one observer can send out a body with
  a certain speed in a certain direction, then any other inertial
  observer can send out a body with this speed in this direction.
\begin{multline*}
\IOb(m)\land
  \IOb(k)\rightarrow\\ \Big[\exists b \big[
  \W(m,b,\vx)\land\W(m,b,\vy)\big] \leftrightarrow \exists
  b\big[ \W(k,b,\vx)\land\W(k,b,\vy)\big] \Big].
\end{multline*}
\end{description}
We call body $b$ {\bf inertial body} iff there is an inertial
  observer $m$ according to who $b$ moves with uniform rectilinear
  motion:
\begin{multline*}
\IB(b) \defiff \exists m\vx\vy\Big[\IOb(m)\land \vx\neq\vy\land
  \W(m,b,\vx)\land \W(m,b,\vy)\land\\\forall \vz\big(
  \W(m,b,\vz)\leftrightarrow \exists \lambda \big[
    \vz=\vx+\lambda(\vy-\vx)\big]\big)\Big].
\end{multline*}

 Let us now formulate the possibility of the existence of
  FTL inertial bodies.
\begin{description}
\item[\underline{\ax{\exists FTLBody}:}] There is an inertial observer
  who can send out an FTL inertial body:
\begin{multline*}
\exists mb\vx\vy \big[\IB(b)\land \IOb(m)\land \W(m,b,\vx)\land
\W(m,b,\vy)\land \\\sqspace(\vx,\vy)>\ls_m^2\cdot \timed(\vx,\vy)^2\big].
\end{multline*}
\end{description}

\ax{\exists FTLBody} implies that inertial observers can send out
a body with arbitrary large speed in any direction if \ax{SpecRel},
\ax{AxThExp^\#}, \ax{CONT} and \ax{AxVel} are assumed:
\begin{proposition}\label{prop-nolim} Let $d\ge 3$.
Assume \ax{SpecRel}, \ax{AxThExp^\#}, \ax{CONT}, \ax{AxVel} and \ax{\exists
  FTLBody}. Then any inertial observer can send out a body with any
speed in any direction:
\begin{equation*}
\IOb(m) \rightarrow \exists b \big[ \W(m,b,\vx)\land \W(m,b,\vy)\big].
\end{equation*}
\end{proposition}
The proof of Prop.\ref{prop-nolim} is in Sect.\ref{sec-proof}.

\section{Hypercomputation in AccRel}
\label{sec-HC}
In this section, we formulate our statement on the logical equivalence
between the existence of FTL signals and the possibility of
hypercomputation in special relativity as a theorem in our first-order
logic language.  To formulate the possibility of hypercomputation as a
formula of our first-order logic language, let us define the {\bf
  life-curve} $\lc_{m}(k)$ of observer $k$ according to observer $m$
as the world-line of $k$ according to $m$ {\it parametrized by the
  time measured by $k$},\label{life-curve} formally:
\begin{equation*}
\lc_{m}(k)\de\{\, \langle t,\vx \rangle\in \Q\times \Q^d \::
\:\exists \vy\enskip k\in \ev_k(\vy)=\ev_m(\vx)\land y_1=t\,\}.
\end{equation*}
 The {\bf range} and {\bf domain}
of a binary relation $R$, is defined as:
\begin{equation*}
  \ran R \de \{\, y : \exists x\enskip R(x,y) \,\} \quad\mbox{and}\quad  \dom R \de \{\, x : \exists y\enskip R(x,y) \,\}.
\end{equation*}

The following formula of our language captures the possibility of
relativistic hypercomputation in the sense used in the theory of
relativistic computation.
\begin{description}\label{HC}
\item[\underline{\ax{HypComp}:}] There are two observers a programmer
  $p$ and a computer $c$ and an instant $\tau$ in the programmer's
  worldline such that the computer has infinite time to compute, and
  during its computation the computer can send a signal $s_t$ to the
  programmer which reaches the programmer before the fixed instant:
\begin{multline*}
  \exists pc\tau \Big[ \Ob(p)\land \Ob(c) \land\forall mx\Big(
      \IOb(m)\land x\ge 0\rightarrow x\in \dom \lc_m(c) \land\\
      \forall t \Big[ t>0\rightarrow\exists t' s_t \big[
      0<t'<\tau \land s_t\in\ev_m\big(\lc_m(c)(t)\big)\cap \ev_m\big(\lc_m(p)(t')\big)\big]\Big]\Big)\Big].
\end{multline*}
\end{description}
 The following axiom ensures the existence of uniformly
accelerated observers.
\begin{description}
\item[\underline{\ax{Ax\exists UnifOb}:}] It is possible to accelerate
  an observer uniformly:\footnote{In relativity theory, uniformly
    accelerated observers are moving along hyperbolas, see, e.g.,
    \cite[\S 3.8, pp.37-38]{dinverno}, \cite[\S 6]{MTW}, \cite[\S
      12.4, pp.267-272]{Rindler}.}
\begin{multline*}
\IOb(m)\rightarrow \exists k \Big[ \Ob(k) \land \dom \lc_m(k)=\Q\\\land
  \forall \vx \big[ \vx\in \ran \lc_m(k) \leftrightarrow
    x_2^2-x_1^2=a^2\land x_3=\ldots=x_d=0\big]\Big].
\end{multline*}
\end{description}

Now we can state our theorem on the logical equivalence between the
existence of FTL signals and the possibility of hypercomputation in
special relativity:
\begin{theorem}\label{thm-hc}
Let $d\ge3$. Then
\begin{equation*}
\{\ax{AccRel},\ax{AxThExp^\#},\ax{Ax\exists Unifob},\ax{AxVel}\}
\models \ax{\exists FTLBody}\leftrightarrow \ax{HypComp} .
\end{equation*}
\end{theorem}
The proof of Thm.\ref{thm-hc} is in Sect.\ref{sec-proof}.

\section{Proofs}
\label{sec-proof}
In this section, we prove Prop.\ref{prop-nolim} and
  Thm.\ref{thm-hc}.

\begin{proof}[Proof of Prop.\ref{prop-nolim}]
Let $m$ be an inertial observer and let $\vx,\vy\in\Q^d$.  By
\ax{AxSymDist}, $\ls_m=1$.  If $\sqspace(\vx,\vy)<\timed(\vx,\vy)^2$,
then there is a body (moreover, an inertial observer) $k$ such that
$\W(m,k,\vx)$ and $\W(m,k,\vy)$ by \ax{AxTheExp^\#}. If
$\sqspace(\vx,\vy)=\timed(\vx,\vy)^2$, then there is a body (moreover,
a light signal) $p$ such that $\W(m,p,\vx)$ and $\W(m,p,\vy)$ by
\ax{AxPh}. So we only have to show that there is a body $b$ such that
$\W(m,b,\vx)$ and $\W(m,b,\vy)$ if
$\sqspace(\vx,\vy)>\timed(\vx,\vy)^2$. By \ax{\exists FTLBody}, there
is an inertial observer who can send out an inertial body $b$
with a certain speed which is faster than the speed of light.  By
  \ax{CONT}, the quantity structure $\langle Q,+,\cdot,\le\rangle$ is
  a real closed field, see \cite[Prop.10.1.2]{SzPhd}. Specially every
  positive number has a square root. Therefore, by \ax{AxThExp^\#}, we
  can rotate the worldview of any observer around the time axis by an
  arbitrary angle; and by Thm.\ref{thm-poi} and axiom \ax{AxThExp^\#},
  there is an inertial observer whose simultaneity is so slanted
  that he sees $b$ moving with speed
  $\sqrt{\sqspace(\vx,\vy)/\timed(\vx,\vy)^2}$. Consequently, there
is an inertial observer who coordinatizes inertial body $b$ moving through
$\vx$ and $\vy$. Then, by \ax{AxVel}, every inertial observer can send
out a body moving through $\vx$ and $\vy$.\qed
\end{proof}

\begin{proof}[Proof of Thm.\ref{thm-hc}]
Assume \ax{AccRel}, \ax{AxThExp^\#}, \ax{Ax\exists Unifob},
\ax{AxVel}, and \ax{\exists FTLBody}. We have to prove \ax{HypComp}.
Let $p$ be an arbitrary inertial observer. Let $\tau=1$. Let $c$ be a
uniformly accelerated observer such that $\W(p,c,\vx)$ iff
$x_2^2-x_1^2=1$ and $x_3=\ldots=x_d=0$, see Fig.\ref{fig-hc}. This
observer $c$ exists and $\dom \lc_p(c)=\Q$ by \ax{Ax\exists UnifOb}
and Prop.\ref{prop-wl} below. By Prop.\ref{prop-sim} below, the simultaneity
of any comoving observer $k$ of $c$ at the event of their meeting goes
through the origin. So by Prop.\ref{prop-nolim}, $k$ any comoving
observer of $c$ (and thus $c$) can send out a body reaching $p$ before
$M=\lc_p(p)(\tau)$ and after $O=\lc_p(p)(0)$, i.e.,
\begin{equation*}
\forall t \Big[ t>0\rightarrow\exists t' s_t \big[ 0<t'<\tau \land
    s_t\in\ev_p\big(\lc_p(c)(t)\big)\cap
    \ev_p\big(\lc_p(p)(t')\big)\big]\Big].
\end{equation*}
Let now $m$ be an arbitrary inertial observer and $x\ge0$.  By
Prop.\ref{prop-sim} below, $\dom \lc_m(c)=\dom\lc_p(c)$. Therefore,
$x\in\dom\lc_m(c)=\Q$. Also by Prop.\ref{prop-wl},
$\ev_m\big(\lc_m(c)(t)\big)=\ev_p\big(\lc_p(c)(t)\big)$ for all
$t\in\Q$.  Therefore,
\begin{equation*}
s_t\in\ev_m\big(\lc_m(c)(t)\big)\cap \ev_m\big(\lc_m(p)(t')\big)
\leftrightarrow s_t\in\ev_p\big(\lc_p(c)(t)\big)\cap
\ev_p\big(\lc_p(p)(t')\big).
\end{equation*}
Consequently,
\begin{equation*}
\forall t \Big[ t>0\rightarrow\exists t' s_t \big[ 0<t'<\tau \land
    s_t\in\ev_m\big(\lc_m(c)(t)\big)\cap
    \ev_m\big(\lc_m(p)(t')\big)\big]\Big].
\end{equation*}
This completes the proof of \ax{HypComp}.

\medskip 
To prove the converse direction, assume that both $\neg \ax{\exists
  FTLBody}$ and \ax{HypComp} hold. Let $p$ and $c$ be arbitrary
observers, and $\tau$ be an arbitrary time instant such that
\ax{HypComp} holds for $p$, $c$ and $\tau$, see
Fig.\ref{fig-nohc}. Since the computer can send a signal to the
programmer during its life from any instant $t>0$ and and there are no
FTL particles, $\wl_m(c)\subseteq
I^-\big(\lc_m(p)(\tau)\big)$\footnote{$I^-(\vx)$ denotes the {\bf
    causal past} of coordinate point $\vx$, i.e., $I^-(\vx)\de\{\, \vy
  \::\: y_1\le x_1\land \timed(\vx,\vy)^2\ge \sqspace(\vx,\vy)\,\}$.}
according to any inertial observer $m$. By the twin paradox theorem,
see, e.g., \cite{twp}, \cite[Thm.7.2.2]{SzPhd}, $c$ maximizes its time
if it moves along a straight line. From this fact, it is easy to see
that the longest path in $I^-\big(\lc_m(p)(\tau)\big)$ starting at
$\lc_m(c)(0)$ is the line segment connecting $\lc_m(c)(0)$ and
$\lc_m(p)(\tau)$. Since even this path is finite, $c$ has only a
finite time to compute. Therefore, subformula $\forall \; x\ge 0
\rightarrow x\in\dom \lc_m(c)$ of \ax{HypComp} cannot be true. This
contradiction proves our statement. \qed\end{proof}

\begin{figure}
\begin{center}
\psfrag{s}[br][br]{FTL signal} 
\psfrag{p}[tl][tl]{Programmer}
\psfrag{po}[tr][tr]{$\lc_m(p)(0)$} 
\psfrag{c}[tl][tl]{Computer} 
\psfrag{t}[br][br]{$\lc_m(p)(\tau)$}
\psfrag{I}[tl][tl]{$I^-\big(\lc_m(p)(\tau)\big)$}
\psfrag{o}[tl][tl]{$\lc_m(c)(0)$} 
\includegraphics[keepaspectratio,width=0.7\textwidth]{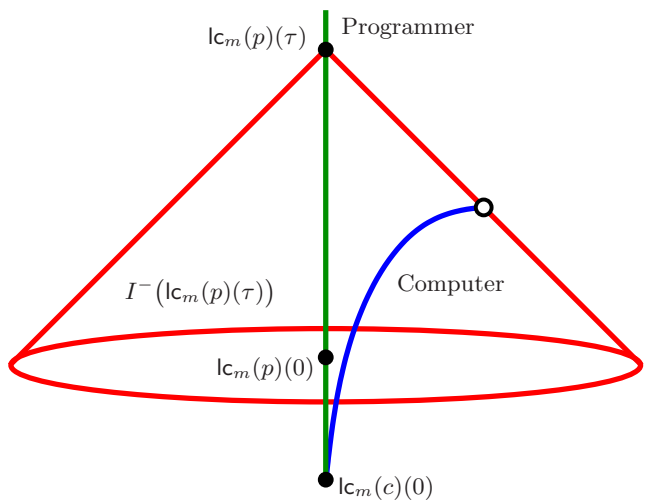}
\caption{Illustration for the proof of Thm.\ref{thm-hc}}
\label{fig-nohc}
\end{center}
\end{figure}

\begin{proposition}\label{prop-sim}
Let $d\ge 3$. Assume \ax{AccRel}.  Let $p$ be an inertial observer and
$c$ be a uniformly accelerated observer such that $\W(p,c,\vx)$ iff
$x_2^2-x_1^2=a^2$ and $x_3=\ldots=x_d=0$ for some $0\neq a\in\Q$. Then
the simultaneity of any comoving inertial observer $k$ of $c$ at $\vr$
through $\vr$, i.e., $\{\vy:\w_{mk}(\vy)_t=\w_{mk}(\vr)_t\}$, contains
the origin $\vo$.
\end{proposition}
\begin{proof}
Let $p$ be an inertial observer, $c$ be a uniformly accelerated
observer, $\vr=\langle r_1,r_2,0,\ldots,0\rangle$ be a point in the
world-line of $c$, and $k$ be an inertial comoving observer of $c$ at
$\vr$. By Thm.\ref{thm-poi}, $\w_{mk}$ is a Poincar\'e
transformation.  Therefore, the simultaneity of $k$ is
Minkowski-orthogonal to his worldline, i.e.,
\begin{multline*}
\forall \vx\vy\vz \Big[\W(m,k,\vx)\land \W(m,k,\vz)\land \vz \neq \vx
  \rightarrow \big[\w_{mk}(\vx)_t=\w_{mk}(\vy)_t\leftrightarrow\\
    (y_1-x_1)(z_1-x_1)=(y_2-x_2)(z_2-x_2)+\ldots+(y_d-x_d)(z_d-x_d)\big]\Big].
\end{multline*}
Therefore, we have to show that line $\vo\vr$ is Minkowski-orthogonal
to $\wl_m(k)$.  By \ax{AxCmv}, the worldline of $k$ is the tangent
line of the worldline of $c$ at $\vr$. Therefore, by Lem.\ref{lem-tan}
below, $\wl_m(k)=\{\vx: r_2x_2=r_1x_1+a^2, x_3=\ldots=x_d=0\}$. Let
$\vz$ be a point of $\wl_m(k)$ different from $\vr$. We have to show
that $(0-r_1)(z_1-r_1)=(0-r_2)(z_2-r_2)$. This equation is the same as
$r_2z_2-r_1z_1=r_2^2-r_1^2$, which follows straightforwardly from
$r_2z_2=r_1z_1+a^2$, i.e., $\vz\in\wl_m(k)$, and $r_2^2=r_1^2+a^2$,
i.e., $\vz\in\wl_m(k)$. Thus line $\vo\vr$ is Minkowski-orthogonal
to $\wl_m(k)$; and this is what we wanted to prove.\qed
\end{proof}

\begin{lemma} \label{lem-tan}
Assume \ax{AxOField} and \ax{CONT}. The tangent line of hyperbola 
\begin{equation*}
\{\vx : x_2^2-x_1^2=a^2, x_3=\ldots=x_d=0\}
\end{equation*}
at its point $\langle
r_1,r_2,0,\ldots,0\rangle$ is 
\begin{equation*}
\{\vx: r_2x_2=r_1x_1+a^2, x_3=\ldots=x_d=0\}.
\end{equation*}
\end{lemma}

\begin{proof} Axioms \ax{AxOField} and \ax{CONT} imply
that $\Q$ is a real closed field, see \cite[Prop.10.1.2]{SzPhd}. By
Tarski's theorem, real closed fields are elementarily equivalent, see
\cite{tarski-dmethod}. Thus something which is expressible in the
language of ordered fields is true in a real closed field iff it is
true in the field of real numbers.  The statement of this lemma can be
formalized in the language of ordered fields and it is straightforward
to show it in the ordered field of real numbers. Therefore, by
Tarski's theorem, the statement is true in every model of
\ax{AxOField} and \ax{CONT}; and this is what we waned to prove.\qed
\end{proof}
The following can be proved about life-curves, see
\cite[Prop. 6.1.6]{SzPhd}.
\begin{proposition}\label{prop-wl}
Let $m$, $k$  and $h$ be observers. Then 
\begin{enumerate}
\item $\wl_m(k)=\ran \lc_m(k)$ if \ax{AxEv^-} is assumed.
\item $\lc_m(k)$ is a function if \ax{AxPh} and \ax{AxSelf^-} are
  assumed and $m$ is an inertial observer.
\item $\dom \lc_m(h)=\dom\lc_k(h)$ and
  $\ev_m\big(\lc_m(h)(t)\big)=\ev_k\big(\lc_k(h)(t)\big)$ holds for
  all $t\in\dom \lc_m(h)$ if $m$ and $k$ are inertial
  observers and \ax{AxEv} is assumed.
\end{enumerate}
\end{proposition}

\section{Concluding Remarks}

We have shown that, in special relativity, the possibility of
hypercomputation is equivalent to the existence of FTL signals. A
natural continuation is to investigate the question concerning the
limits of the possibility of using FTL particles in hypercomputation
in special and general relativity theories. For example, is there a natural
assumption on spacetime which does not forbid the existence of FTL
particles, but makes it impossible to use them for hypercomputation?

Of course our construction contains several engineering
  difficulties.  For example, the larger the distance the more
  difficult to aim with a signal. Therefore, the computer has to
  calculate the speed of the FTL signal more and more accurately to
  ensure that the signal arrives to the programmer between events $O$
  and $M$, see Fig.\ref{fig-hc}. Thus the computer has to be able to
  aim with the FTL signal with arbitrary precision.

\end{document}